\documentclass{llncs}
\usepackage{rotating}
\usepackage{amsmath}
\usepackage{amsfonts}
\usepackage{amssymb}
\usepackage{graphics}

\begin{document}
\frontmatter          


\pagestyle{headings}  

\mainmatter
\title{Constructing a Knowledge Base for Gene Regulatory Dynamics by Formal Concept Analysis Methods}
\titlerunning{Constructing a Gene Regulatory Knowledge Base}
\author{Johannes Wollbold\inst{1} \inst{2} \and
Reinhard Guthke\inst{2} \and Bernhard Ganter\inst{1}}
\authorrunning{Johannes Wollbold et al.}
\institute{University of Technology, Institute of Algebra, Dresden, Germany\\
\texttt{http://www.math.tu-dresden.de/alg/algebra.html}\\
\email{jwollbold@gmx.de} \and Leibniz Institute for Natural Product Research and Infection Biology
- Hans-Kn\"oll-Institute (HKI) Jena, Germany}
\maketitle

\begin{abstract}
Our aim is to build a set of rules, such that reasoning over temporal dependencies within gene
regulatory networks is possible. The underlying transitions may be obtained by discretizing
observed time series, or they are generated based on existing knowledge, e.g. by Boolean networks
or their nondeterministic generalization. We use the mathematical discipline of formal concept
analysis (FCA), which has been applied successfully in domains as knowledge representation, data
mining or software engineering. By the \textit{attribute exploration} algorithm, an expert or a
supporting computer program is enabled to decide about the validity of a minimal set of
implications and thus to construct a sound and complete knowledge base. From this all valid
implications are derivable that relate to the selected properties of a set of genes. We present
results of our method for the initiation of sporulation in \textit{Bacillus subtilis}. However the
formal structures are exhibited in a most general manner. Therefore the approach may be adapted to
signal transduction or metabolic networks, as well as to discrete temporal transitions in many
biological and nonbiological areas.
\end{abstract}

\section{Introduction}
As the mathematical methodology of formal concept analysis (FCA) is little known within systems
biology, we give a short overview of its history and purposes. During the early years 1980, FCA
emerged within the community of set and order theorists, algebraists and discrete mathematicians.
Its first aim was to find a new, concrete and meaningful approach to the understanding of complete
lattices (ordered sets such that for every subset the supremum and the infimum exist). The
following discovery showed to be very fruitful: Every complete lattice is representable as a
hierarchy of concepts, which were conceived as sets of objects sharing a maximal set of attributes.
This paved the way for using the developed field of lattice theory for a transparent and complete
representation of very different types of knowledge. FCA was inspired by the pedagogue Hartmut von
Hentig \cite{Hen74} and his program of restructuring sciences, with a view to interdisciplinary
collaboration and democratic control. The philosophical background goes back to Charles S. Peirce
(1839 - 1914), who condensed some of his main ideas to the pragmatic maxim:
\begin{center}
\parbox{10cm}{\emph{Consider what effects, that might conceivably have practical bearings, we
conceive the objects of our conception to have. Then, our conception of these effects is the whole
of our conception of the object.} \cite[5.402]{Pei35}}
\end{center}

In that tradition, FCA aims at unfolding the observable, elementary properties defining the objects
subsumed by scientific concepts. If applied to temporal transitions, effects of homogeneous classes
of states can be modeled and predicted in a clear and concise manner. Thus FCA seems to be
appropriate to describe causality - and the limits of its understanding.

At present, FCA is a richly developed mathematical theory, and there are practical applications in
various fields as data and text mining, knowledge management, semantic web, software engineering or
economics \cite{GSW05}. FCA has been used for the analysis of gene expression data in \cite{Cho06}
and \cite{Mot08}, but this is the first approach of applying it to model (gene) regulatory
networks. The mathematical framework of FCA is very general and open, such that multifarious
refinements are possible, according to current approaches of modeling dynamics within systems
biology. On the other hand, we developed a formal structure for general discrete temporal
transitions. They occur in a variety of domains: control of engineering processes, development of
the values of variables or objects in a computer program, change of interactions in social
networks, a piece of music, etc.

In this paper, however, the examples are uniquely biological. The purpose is to construct a
knowledge base for reasoning about temporal dependencies within gene regulatory or signal
transduction networks, by the \textit{attribute exploration} algorithm: For a given set of
interesting properties, it builds a sound, complete and nonredundant knowledge base. This minimal
set of rules has to be checked by an expert or a computer program, e.g. by comparison of knowledge
based predictions with data.

Since there exist relatively fixed thresholds of activation for many genes, it is a common
abstraction to consider only two expression levels \textit{off} and \textit{on}. The classical
approach of Boolean networks \cite{Kau93} is able to capture essential dynamic aspects of
regulatory networks. Our present work is based on it, which also makes it possible to use standard
mathematical and logical derivations for deciding many rules automatically and for a scaling up to
larger networks. Nevertheless, the introduction of more fine-grained expression levels remains
possible, e.g. in the sense of qualitative reasoning \cite{KGC05}. Further, this work is influenced
by computation tree logic \cite{Cha04}, automata theory and a FCA modeling of temporal transitions
in \cite{Rud01}. Temporal concept analysis as developed by K.E. Wolff \cite[p. 127-148]{GSW05} is
more directed toward a structured visualization of experimental time series then toward temporal
logic. We applied it to the analysis of gene expression data in \cite{Wol08a}.

In Section \ref{methods}, the general mathematical framework will be developed. Section
\ref{sec:results} gives results for a \textit{B. subtilis} Boolean network. In Section
\ref{conclusion}, we will discuss the potential of the method and make some proposals for improving
it by solving mathematical problems which have emerged from the applications.

\section{Methods}\label{methods}
\subsection{Fundamental Structures of Formal Concept Analysis}
One of the classical aims of FCA is the structured, compact but complete visualization of a data
set by a conceptual hierarchy. We briefly introduce its basic definitions; for an easy example see
\texttt{http://www.upriss.org.uk/fca/fca.html}.

\begin{definition}\label{def:context}
A \textbf{formal context} $(G,M,I)$ defines a relation $I \subseteq G \times M$ between objects
from a set $G$ and attributes from a set $M$. The set of the attributes common to all objects in $A
\subseteq G$ is denoted by the $'$-operator:
\[A' := \{m \in M | \: gIm \text{ for all } g \in A\}.\]
The set of the objects sharing all attributes in $B \subseteq M$ is
\[B' := \{g \in G | \: gIm \text{ for all } m \in B\}.\]
\end{definition}

\begin{definition}\label{def:concept}
A \textbf{formal concept} of the context $(G,M,I)$ is a pair $(A,B)$ with $A \subseteq G, \:
B\subseteq M,\: A'=B$ and $B'=A$. $A$ is the \textbf{extent}, $B$ the \textbf{intent} of the
concept $(A,B)$.
\end{definition}

Thus a formal context $(G,M,I)$ is a special, but universally applicable type of a data table,
provided with two operators $\mathcal{C}_M: \mathcal{P}(M) \rightarrow \mathcal{P}(M), \:B
\subseteq M \mapsto B''$ and $\mathcal{C}_G: \mathcal{P}(G) \rightarrow \mathcal{P}(G), \:A \subset
G \mapsto A''$. It is easy to see that they are closure operators, with the properties monotony,
extension and idempotency \cite[Definition 14]{GW99}. It follows that the set of all extents resp.
intents of a formal context is a closure system, i.e. it is closed under intersections
\cite[Theorem 1]{GW99}.

Formal concepts can be ordered by set inclusion of the extents or - dually, with the inverse order
relation - of the intents. With this order, the set of all concepts of a given formal context is a
complete lattice \cite[Theorem 3]{GW99} (Figure \ref{fig:noStress}).

During the interactive \textit{attribute exploration} algorithm \cite[p. 85ff.]{GW99}, an expert is
asked about the general validity of basic implications $A \rightarrow B$ between the attributes of
a given formal context $(G,M,I)$. An implication has the meaning: "If an object $g \in G$ has all
attributes $a \in A \subseteq M$, then it has also all attributes $b \in B \subseteq M$." If the
expert denies, s/he must provide a counterexample, i.e. a new object of the context. If s/he
accepts, the implication is added to the \textit{stem base} of the - possibly enlarged - context. A
theorem by Duquenne-Guiges \cite[Theorem 8]{GW99} ensures that every implication semantically valid
in the underlying formal context can be derived syntactically from this minimal set by the
Armstrong rules \cite[Proposition 21]{GW99}. In many applications, one is merely interested in the
implicational logic of a given formal context, and there is no need for an expert to confirm the
implications.

\subsection{Constructing the Knowledge Base - Summary of the Method}
We start with two sets:
\begin{itemize}
\item The universe $E$. The elements of $E$ represent the \textbf{entities} of the world which we
are interested in.
\item The set $F$ (\textbf{fluents}) denotes changing properties of the entities.
\end{itemize}

A \textbf{state} $\varphi \in G$ is an assignment of values in $F$ to the variables $e\in E$, hence
it is defined by a specific choice of attributes $m \in M \subseteq E \times F$.\footnote{Thus - as
usual - states with the same variable values are identified. It would also be possible to
distinguish them as situations by introducing a new attribute, e.g. "time interval".} By means of a
state context (Definition \ref{def:stateCxt}, Table \ref{tab:extCxt} left part), temporal data can
be translated into the language of FCA. The dynamics is modeled by a binary relation $R \subseteq G
\times G$ on the set of states, which gives rise to a transition context $\mathbb{K}$ (Definition
\ref{def:transCxt}, Table \ref{tab:trCxt}): the objects are transitions (elements of the relation)
and the attributes the values of the entities defining the input and output state of a transition.
This data table may reflect observations repeated at different time points, or the transitions may
be generated by a dynamic model. As to the latter, we are focusing here on Boolean networks, i.e
sets of Boolean functions for each entity (Definition \ref{BoolNet}).

It is promising to consider the transitive closure of $R$. Objects of the transitive context
$\mathbb{K}_t$ (Definition \ref{def:transitCxt}, Table \ref{tab:transitCxt}) then are pairs of
states such that the output state emerges from the input state by some transition sequence of
arbitrary length. Finally we extend the state context $\mathbb{K}_s$ by the temporal attributes
always($m$), eventually($m$) and never($m$), which are determined by the given transitions
(Definition \ref{def:stateCxt}, Table \ref{tab:extCxt}).

The defined mathematical structures may be used in various ways. For instance, one could evaluate -
i.e. generalize implications or reject them supposing outliers or by reason of special conditions -
experimental time series by comparison with existing knowledge. Our general procedure is the
following:
\begin{enumerate}
\item\label{item:Kobs} Discretize a set of time series of gene expression measurements and
transform it to an observed transition context $\mathbb{K}^{obs}$.
\item\label{item:BoolNet} For a set of interesting genes, translate interactions from biological
literature and databases into a Boolean network.
\item\label{item:K} Construct the transition context $\mathbb{K}$ by a simulation
 starting from a set of  states,
 e.g. the initial states of $\mathbb{K}^{obs}$ or all states (for small networks).
\item\label{item:K_t} Derive the respective transitive contexts $\mathbb{K}_t$ and
$\mathbb{K}^{obs}_t$.
\item \label{item:compObs}Perform attribute exploration of $\mathbb{K}_t$. Decide about an
implication $A\rightarrow B$, $A,B \subseteq M$, by checking its validity in $\mathbb{K}^{obs}_t$
and/or by searching for supplementary knowledge. Possibly provide a counterexample from
$\mathbb{K}^{obs}_t$.
\item\label{item:query} Answer queries from the modified context $\mathbb{K}_t$ and from its stem base.
\end{enumerate}
In step \ref{item:compObs}, automatic decision criteria could be thresholds of support $q= |(A
\cup B)'|$ and confidence $p= \frac{|(A \cup B)'|}{|A'|}$ for an implication in
$\mathbb{K}^{obs}_t$. A weak criterion is to reject only implications with support 0 (but if no
object in $\mathbb{K}^{obs}_t$ has all attributes from A, the implication is not violated). In
\cite{Wol07}, a strong criterion has been applied: implications of $\mathbb{K}_t$ had to be valid
also in the observed context ($p=1$). This is equivalent to an exploration of the union of the two
contexts.

In Section \ref{sec:results} we will analyse pure knowledge based simulations; the validation by
data and experimental literature had been done before in \cite{deJ04}. For that reason, in step
\ref{item:compObs}. the stem base is computed automatically, without further confirmation by an
expert.

Step \ref{item:BoolNet}. could be supported by text mining software. Then attribute exploration
provides strong criteria of validation. We implemented the steps \ref{item:Kobs}., \ref{item:K}.
and \ref{item:K_t}. in R [www.r-project.org]. For step \ref{item:compObs}., we used the Java tool
Concept Explorer [http://sourceforge.net/projects/conexp]. The output was translated with R into a
PROLOG knowledge base. The R scripts are available on request.

\subsection{Definition of the Relevant Formal Contexts}\label{subsec:trans}
With given sets $E$ and $F$, we define a state as a map $\varphi: E \rightarrow F$. To explore
static features of states, the following formal context\footnote{It is equivalent to a \textit{many-valued
context} with \textit{nominal scale} \cite[Section 1.3]{GW99}, \cite[p. 123]{Wol07}. For better
readability, we draw the contexts in the latter form (Table \ref{tab:extCxt}). Deriving a
one-valued context according to Definition \ref{def:stateCxt} is obvious: each many-valued
attribute $e$ is replaced by $\{(e,f)|\: f \in F\})$, e.g. SigA by SigA.off and SigA.on. If an
attribute $e$ takes exactly one of these values, negation of \textit{on} and \textit{off} is
expressed. Other kinds of scaling like the \textit{interordinal scale} could be interesting, if
there are more than two levels ($|F| > 2$).} is defined.
\begin{definition}\label{def:stateCxt}
Given two sets $E$ (entities) and $F$ (fluents), a \textbf{state context} is a formal context
$(G,M,I)$ with $G\subseteq F^E := \{\varphi: E \rightarrow F\}$ and $M \subseteq E \times F$; its
relation $I$ is given as $\varphi\, I\,(e,f) \Leftrightarrow \varphi(e) = f$, for all $\varphi \in
G, e \in E$ and $f \in F$.
\end{definition}

\begin{definition}\label{def:transCxt}
Given a state context $(G,M,I)$ and a relation $R \subseteq G\times G$, a \textbf{transition
context} $\mathbb{K}$ is the context $(R, M\times\{in,out\}, \tilde{I})$ with the property
\begin{equation}\label{trContext}
\forall i \in \{in,out\}\!:\: (\varphi^{in},\varphi^{out}) \tilde{I} (e,f,i) \Leftrightarrow \varphi^i(e)=f.
\end{equation}
\end{definition}
\begin{table}\caption{A transition context for the states of Table \ref{tab:extCxt}, with all attributes
that are changing during the small simulation, as well as Spo0A and Spo0AP.}\centering
\begin{tabular}{|c|ccccccc|ccccccc|}
\hline\textbf{Transition} &\begin{sideways}\textbf{KinA$^{in}$}\end{sideways}
&\begin{sideways}\textbf{Spo0A$^{in}$}\end{sideways}
&\begin{sideways}\textbf{Spo0AP$^{in}$}\end{sideways}
&\begin{sideways}\textbf{AbrB$^{in}$}\end{sideways}
&\begin{sideways}\textbf{Spo0E$^{in}$}\end{sideways}
&\begin{sideways}\textbf{SigH$^{in}$}\end{sideways}
&\begin{sideways}\textbf{Hpr$^{in}$}\end{sideways}
&\begin{sideways}\textbf{KinA$^{out}$}\end{sideways}
&\begin{sideways}\textbf{Spo0A$^{out}$}\end{sideways}
&\begin{sideways}\textbf{Spo0AP$^{out}$}\end{sideways}
&\begin{sideways}\textbf{AbrB$^{out}$}\end{sideways}
&\begin{sideways}\textbf{Spo0E$^{out}$}\end{sideways}
&\begin{sideways}\textbf{SigH$^{out}$}\end{sideways}
&\begin{sideways}\textbf{Hpr$^{out}$}\end{sideways}\\
\hline
$(\varphi_0^{in}, \varphi_1^{out})$ &-&+&-&-&-&-&+&-&+&-&+&+&+&-\\
$(\varphi_1^{in}, \varphi_2^{out})$ &-&+&-&+&+&+&-&+&+&-&-&-&-&+\\
$(\varphi_2^{in}, \varphi_1^{out})$ &+&+&-&-&-&-&+&-&+&-&+&+&+&-\\
\hline
\end{tabular}
\label{tab:trCxt}
\end{table}
Transitions may be generated by a Boolean network:
\begin{definition}\label{BoolNet}
Let $E$ be an arbritray set of entities, $F:=\{0,1\}$ (fluents), and states $G \subseteq F^E$. Then
a transition function $F^E \rightarrow F^E$ is called a \textbf{Boolean network}.
\end{definition}
We will identify the elements of $F$ with $-, +$ or \textit{off, on} respectively. This definition
is subsumed by the definition of a dynamic network in \cite[Definition p. 34]{Lau05}, with a set of
variables $E$ and state sets $X_1 = ... = X_n = F$. We use a parallel update schedule, i.e. the
order relation on $E$ is empty. Boolean networks may be generalized in order to include
nondeterminism; then different output states $\varphi^{out}$ are generated from a single input
state $\varphi^{in}$ (see Section \ref{sec:results}, compare \cite{Wol07}).
\begin{definition}\label{def:transitCxt}
A transition context $\mathbb{K}$ with a transitively closed relation $t(R) \subseteq G \times G$
is called a \textbf{transitive context} $\mathbb{K}_t$.
\end{definition}
\begin{table}\caption{The transitive context derived from the transition context of Table
\ref{tab:trCxt}.}\centering
\begin{tabular}{|c|ccccccc|ccccccc|}
\hline\textbf{Transition} &\begin{sideways}\textbf{KinA$^{in}$}\end{sideways}
&\begin{sideways}\textbf{Spo0A$^{in}$}\end{sideways}
&\begin{sideways}\textbf{Spo0AP$^{in}$}\end{sideways}
&\begin{sideways}\textbf{AbrB$^{in}$}\end{sideways}
&\begin{sideways}\textbf{Spo0E$^{in}$}\end{sideways}
&\begin{sideways}\textbf{SigH$^{in}$}\end{sideways}
&\begin{sideways}\textbf{Hpr$^{in}$}\end{sideways}
&\begin{sideways}\textbf{KinA$^{out}$}\end{sideways}
&\begin{sideways}\textbf{Spo0A$^{out}$}\end{sideways}
&\begin{sideways}\textbf{Spo0AP$^{out}$}\end{sideways}
&\begin{sideways}\textbf{AbrB$^{out}$}\end{sideways}
&\begin{sideways}\textbf{Spo0E$^{out}$}\end{sideways}
&\begin{sideways}\textbf{SigH$^{out}$}\end{sideways}
&\begin{sideways}\textbf{Hpr$^{out}$}\end{sideways}\\
\hline
$(\varphi_0^{in}, \varphi_1^{out})$ &-&+&-&-&-&-&+&-&+&-&+&+&+&-\\
$(\varphi_0^{in}, \varphi_2^{out})$ &-&+&-&-&-&-&+&+&+&-&-&-&-&+\\
$(\varphi_1^{in}, \varphi_1^{out})$ &-&+&-&+&+&+&-&-&+&-&+&+&+&-\\
$(\varphi_1^{in}, \varphi_2^{out})$ &-&+&-&+&+&+&-&+&+&-&-&-&-&+\\
$(\varphi_2^{in}, \varphi_1^{out})$ &+&+&-&-&-&-&+&-&+&-&+&+&+&-\\
$(\varphi_2^{in}, \varphi_2^{out})$ &+&+&-&-&-&-&+&+&+&-&-&-&-&+\\
\hline
\end{tabular}
\label{tab:transitCxt}
\end{table}
\begin{definition} A \textbf{state context} $\mathbb{K}=(G,M,I)$ is \textbf{extended} to a
formal context $\mathbb{K}_s=(G, M \cup T, I \cup I_T)$ by a set of temporal attributes $T :=
\{always(m)|\: m \in M\} \cup \{never(m)|\: m \in M\} \cup \{eventually(m)|\: m \in M\}$. Let
$\varphi^{in} \in G, m=(e,f),\: e \in E,\: f \in F$, and $t(R) \subseteq G \times G$ a transitively
closed relation. The relation $I_T$ of $\mathbb{K}_s$ then is defined as follows:
\begin{align*}
\varphi^{in}\: I_T\:always(m) \Leftrightarrow\: &\forall (\varphi^{in}, \varphi^{out}) \in t(R): \varphi^{out}(e) = f\\
\varphi^{in}\: I_T\:never(m) \Leftrightarrow\: &\forall (\varphi^{in}, \varphi^{out}) \in t(R): \varphi^{out}(e) \neq f\\
\varphi^{in}\: I_T\:eventually(m) \Leftrightarrow\: &\exists (\varphi^{in}, \varphi^{out}) \in t(R): \varphi^{out}(e) = f
\end{align*}
For $B \subseteq T$, set $always(B)$ := $\{\{always(b_1),..., always(b_i)\} \mid b_1,...,b_i \in
B\}$, and analogously $never(B)$ and $eventually(B)$.
\end{definition}
The attributes will be abbreviated to alw($m$), nev($m$) and ev($m$). In a nondeterministic
setting, alw($m$) and nev($m$) refer to all possible transition paths starting from $\varphi^{in}$,
ev($m$) to the existence of a path.
\begin{table}\caption{\textit{Left part}: A state context corresponding to a simulation starting from a \textit{B.
subtilis} state without nutritional stress (see Section \ref{simNoStress}, \cite[Table 4]{Ste07}).
+: on, -: off. \textit{Right part}: extension by temporal attributes. Here they are the same for
all states, since these reach the \textit{attractor} (limit state cycle) $\{\varphi_1, \varphi_2\}$
after at most one time step.}\centering
\begin{tabular}{|c|ccccccc|ccccccc|}
\hline\textbf{State} &\begin{sideways}\textbf{KinA$^{in}$}\end{sideways}
&\begin{sideways}\textbf{Spo0A$^{in}$}\end{sideways}
&\begin{sideways}\textbf{Spo0AP$^{in}$}\end{sideways}
&\begin{sideways}\textbf{AbrB$^{in}$}\end{sideways}
&\begin{sideways}\textbf{Spo0E$^{in}$}\end{sideways}
&\begin{sideways}\textbf{SigH$^{in}$}\end{sideways}
&\begin{sideways}\textbf{Hpr$^{in}$}\end{sideways} &\begin{sideways}\textbf{ev(KinA)}\end{sideways}
&\begin{sideways}\textbf{alw(KinA)}\end{sideways}
&\begin{sideways}\textbf{nev(Spo0AP)}\end{sideways}
&\begin{sideways}\textbf{ev(AbrB)}\end{sideways} &\begin{sideways}\textbf{alw(AbrB)}\end{sideways}
&\begin{sideways}\textbf{ev(Hpr)}\end{sideways}
&...\\
\hline
$\varphi_0^{in}$ &-&+&-&-&-&-&+&x&&x&x&&x&\\
$\varphi_1^{in}$ &-&+&-&+&+&+&-&x&&x&x&&x&\\
$\varphi_2^{in}$ &+&+&-&-&-&-&+&x&&x&x&&x&\\
\hline
\end{tabular}
\label{tab:extCxt}
\end{table}

\subsection{Dependency of Contexts and Background Knowledge}\label{subsec:bg}
In the following we will present first mathematical results that can improve computability; they
are not necessary for the understanding of the application in Section \ref{sec:results}. By
entering background knowledge (not necessarily implications) prior to an attribute exploration, the
algorithm may be shortened considerably \cite[p. 101-113]{GSW05}. We searched for first order logic
background formula in order to use the results of an attribute exploration for the exploration of
the next context in the hierarchy. Then the implications of the latter context are derivable from
this background knowledge and a reduced set of new implications. Also during the exploration of one
context, implications can be decided automatically based on already accepted implications. In this
way the expert is enabled to concentrate on really interesting hypotheses. Thus, the implications
of a state context hold in the input and output part of the corresponding transition context (for
an example see \cite[p. 149f.]{Rud01}). Related to transitive and extended state contexts, the
subsequent result holds:
\begin{proposition}\label{prop:KtKs}
Let $\mathbb{K}_s=(G, M \cup T, I \cup I_T)$ an extended state context. Suppose the relation $t(R)
\subseteq G \times G$ is the object set of the transitive context $\mathbb{K}_t = (t(R),
M\times\{in,out\}, \tilde{I})$. Then the following entailments between implications of both
contexts are valid:
\begin{align}
B^{in} \rightarrow m^{out}\text{ in }\mathbb{K}_t &\equiv B \rightarrow \text{always}(m)
\text{ in }\mathbb{K}_s\label{prop:1}\\
\label{prop:2}B^{in} \rightarrow m^{out}\text{ in }\mathbb{K}_t &\models B \rightarrow
\text{eventually}(m) \text{ in }\mathbb{K}_s\\
\label{prop:3}B^{out} \rightarrow m^{out} \text{ in }\mathbb{K}_t &\models \text{always(B)}
\rightarrow \text{always(m)} \text{ in }\mathbb{K}_s\\
\label{prop:4}B^{out} \rightarrow m^{out} \text{ in }\mathbb{K}_t &\models \text{eventually(B)}
\rightarrow \text{eventually}(m) \text{ in}\: \mathbb{K}_s\\
\label{prop:5}B^{in}\cup m^{out} \rightarrow \bot\text{ in }\mathbb{K}_t &\models B \rightarrow
\text{never}(m) \text{ in }\mathbb{K}_s
\end{align}
If the latter implication does not follow from the stem base of $\mathbb{K}_t$, this is equivalent
to $B \rightarrow \text{eventually}(m)$ in $\mathbb{K}_s$.
\end{proposition}
\begin{proof}
The proofs are straightforward from the definitions. \qed
\end{proof}

In order to get a complete overview on valid entailments, as a first step we performed rule
exploration \cite{Zic91} of the following \textbf{test context}, i.e. exploration of Horn rules
instead of implications, thus variables are admitted: The objects are all possible $\mathbb{K}_t$
respectively the corresponding $\mathbb{K}_s$, and the attributes are the following classes of
implications with "homogeneous" premises. Then the explored rules for implications correspond to
entailments valid for the semantics given by the objects, the transitive contexts. The sets are
nonempty subsets of $M$, $m=(e,f),\: f\in F :=\{0,1\}$, and $m \in B_0, C_0\:(B_1,C_1) \Rightarrow
\varphi(e) =0\: (1)$.\label{imp} We suppose that all states and transitions are completely defined.
\begin{center}
\begin{minipage}[t]{60mm}
\begin{enumerate}
	\item $B^{in}\rightarrow C^{in}$
  \item $B^{in}\rightarrow C_0^{out} \equiv B^{in}\rightarrow$ nev($C_1$)
  \item $B^{in}\rightarrow C_1^{out} \equiv B^{in}\rightarrow$ alw($C_1$)
  \item $B^{in}\rightarrow$ ev($C_1$)
  \item $B_0^{out}\rightarrow C^{in}$
  \item $B_1^{out}\rightarrow C^{in}$
  \item ev($B_1$) $\rightarrow C^{in}$
  \item alw($B_1$) $\rightarrow C^{in}$
  \item nev($B_1$) $\rightarrow C^{in}$
  \item $B_0^{out}\rightarrow C_0^{out}$
  \item $B_0^{out}\rightarrow C_1^{out}$
\end{enumerate}
\end{minipage}
\begin{minipage}[t]{45mm}
 \begin{enumerate}\setcounter{enumi}{11}
 \item $B_1^{out}\rightarrow C_0^{out}$
  \item $B_1^{out}\rightarrow C_1^{out}$
  \item ev($B_1$) $\rightarrow$ ev($C_1$)
  \item ev($B_1$) $\rightarrow$ alw($C_1$)
  \item ev($B_1$) $\rightarrow$ nev($C_1$)
  \item alw($B_1$) $\rightarrow$ ev($C_1$)
  \item alw($B_1$) $\rightarrow$ alw($C_1$)
  \item alw($B_1$) $\rightarrow$ nev($C_1$)
  \item nev($B_1$) $\rightarrow$ ev($C_1$)
  \item nev($B_1$) $\rightarrow$ alw($C_1$)
  \item nev($B_1$) $\rightarrow$ nev($C_1$)
\end{enumerate}
\end{minipage}
\end{center}

The equivalences in 2. and 3. follow from Proposition \ref{prop:KtKs}(\ref{prop:1}). Since the
implications comprising input attributes are independent from those related only to output
attributes, rule exploration was performed (almost) independently for the first 9 and the remaining
13 implications. Results for the second part are shown here.

The exploration started from a hypothetical context as single object of the test context, where no
implications were valid. Before, we had added 25 known entailments as background rules (BR), like
those of Proposition \ref{prop:KtKs} or following from the definitions, like alw($B_1$)
$\rightarrow$ ev($B_1$). A counterexample represents a significant class of contexts. They had to
be chosen carefully, since an object not having its maximal attribute set might preclude a valid
entailment.\footnote{Thus the attribute set of a counterexample must be a concept intent in the
final test context.} The exploration resulted in the following stem base of only 14 entailments.
Most of them are background rules (they are accepted automatically during the exploration), but not
all of these are needed in order to derive all valid entailments between the chosen implications.
This demonstrates the effectivity and minimality of the algorithm. Entailments 5., 6., 7. and 10.
were newly found.
\begin{enumerate}
  \item nev($B_1$) $\rightarrow$ alw($C_1$)
  $\models$ nev($B_1$) $\rightarrow$ ev($C_1$) (BR 1)
  \item nev($B_1$) $\rightarrow$ ev($C_1$), nev($B_1$) $\rightarrow$ nev($C_1$)
  $\models \bot$ (BR 11)
  \item alw($B_1$) $\rightarrow$ alw($C_1$)
  $\models$ alw($B_1$) $\rightarrow$ ev($C_1$) (BR 2)
  \item alw($B_1$) $\rightarrow$ ev($C_1$), alw($B_1$) $\rightarrow$ nev($C_1$)
  $\models \bot$ (BR 14)
  \item ev($B_1$) $\rightarrow$ nev($C_1$), nev($B_1$) $\rightarrow$ nev($C_1$)
  $\models$  $B^{in} \rightarrow C_0^{out}$, $B_0^{out}$ $\rightarrow$ $C_0^{out}$, $B_1^{out}$ $\rightarrow$ $C_0^{out}$
  \item ev($B_1$) $\rightarrow$ nev($C_1$), alw($B_1$) $\rightarrow$ nev($C_1$)
  $\models$ $B_1^{out}$ $\rightarrow$ $C_0^{out}$
  \item ev($B_1$) $\rightarrow$ nev($C_1$), alw($B_1$) $\rightarrow$ ev($C_1$)
  $\models \bot$
  \item ev($B_1$) $\rightarrow$ alw($C_1$) $\models$
  ev($B_1$) $\rightarrow$ ev($C_1$), alw($B_1$) $\rightarrow$ ev($C_1$),
  alw($B_1$) $\rightarrow$ alw($C_1$) (BR 3)
  \item ev($B_1$) $\rightarrow$ ev($C_1$)
  $\models$ alw($B_1$) $\rightarrow$ ev($C_1$) (BR 4)
  \item ev($B_1$) $\rightarrow$ ev($C_1$), nev($B_1$) $\rightarrow$ ev($C_1$)
  $\models$ $B^{in}$ $\rightarrow$ $C_1^{out}$, $B_0^{out}$ $\rightarrow$ $C_1^{out}$,
  $B_1^{out}$ $\rightarrow$ $C_1^{out}$, ev($B_1$) $\rightarrow$ alw($C_1$)
  \item $B_1^{out}$ $\rightarrow$ $C_1^{out}$  $\models$ ev($B_1$) $\rightarrow$ ev($C_1$),
  alw($B_1$) $\rightarrow$ ev($C_1$), alw($B_1$) $\rightarrow$ alw($C_1$) (BR 4, BR
  5 $\Leftarrow$ Proposition \ref{prop:KtKs}(\ref{prop:3})(\ref{prop:4}))
  \item $B_1^{out}$ $\rightarrow$ $C_0^{out}$
  $\models$ alw($B_1$) $\rightarrow$ nev($C_1$) (BR 9 $\Leftarrow$ Proposition \ref{prop:KtKs}(\ref{prop:3}))
  \item $B_0^{out}$ $\rightarrow$ $C_1^{out}$ $\models$
  nev($B_1$) $\rightarrow$ ev($C_1$), nev($B_1$) $\rightarrow$ alw($C_1$)
  (BR 1, 10 $\Leftarrow$ Proposition \ref{prop:KtKs}(\ref{prop:3}))
  \item $B_0^{out}$ $\rightarrow$ $C_0^{out}$
  $\models$ nev($B_1$) $\rightarrow$ nev($C_1$) (BR 6 $\Leftarrow$ Proposition \ref{prop:KtKs}(\ref{prop:3}))
\end{enumerate}

It remains to prove the rules of this stem base, which is easy; we are giving some hints: BR 1, 2,
3, and 4 are based on alw($A$) $\rightarrow$ ev($A$), $A \subseteq M$, and BR 11 and 14 on
nev($C_1$), ev($C_1$) $\rightarrow \bot$ ($\bot$ = set of all attributes, and the corresponding
object set is empty).

7.: Since nev($C_1$) and ev($C_1$) do not occur together by definition, the combination of the two
implications has support 0 in the test context. In the underlying contexts, the premise alw($B_1$)
(a subcase of ev($B_1$)) is no attribute of any state. The implication alw$(B_1)\rightarrow \bot$
holds, which has not been considered explicitly.

10.: Inversely, in all possible cases the states / transitions have the attribute ev($C_1$), and
therefore also alw($C_1$) and $C_1^{out}$. Explicitly: $\top \rightarrow$ ev($C_1$), $\top
\rightarrow$ alw($C_1$), $\top \rightarrow$ $C_1^{out}$. 5. is a parallel rule concerning
nev($C_1$).

Rules 7. and 10. suggest that implications with empty premise $\top$ or conclusion $\bot$ should be
considered explicitly. If the counterexamples have maximal attribute sets, as a conclusion it can
be stated that we have derived a set of rules representing a minimal, sound and complete entailment
calculus for the selected classes of implications for transition and state contexts.

\section{Results: Sporulation in \textit{Bacillus subtilis}}\label{sec:results}
In order to demonstrate the characteristics of the proposed method, we will apply it to a gene
regulatory network assembled in \cite{deJ04} and transformed to a Petri net as well as a Boolean
network in \cite{Ste07}.

\textit{B. subtilis} is a gram positive soil bacterium. Under extreme environmental stress, it
produces a single endospore, which can survive ultraviolet or gamma radiation, acid, hours of
boiling or long periods of starvation. The bacterium leaves the vegetative growth phase in favour
of a dramatically changed and highly energy consuming behaviour, and it dies at the end of the
sporulation process. This corresponds to a switch between two clearly distinguished genetic
programs, which are complex but comparatively well understood.

By literature and database search, de Jong et al. \cite{deJ04} identified 12 main regulators,
constructed a model of piecewise linear differential equations and obtained realistic simulation
results. An exogenous signal (starvation) triggers the phosphorylation of the transcription factor
Spo0A to Spo0AP by the kinase KinA; this process is reversible by the phosphatase Spo0E. Spo0AP is
necessary to transcribe SigF, which regulates genes initiating sporulation and therefore is an
output of the model. The interplay with other transcription factors AbrB, Hpr, SigA, SigF, SigH and
SinR is graphically represented in \cite[Figure 3]{deJ04}; SinI inactivates SinR by binding to it.
SigA and Signal are considered as an input to the model and are always on. Table \ref{BoolEq} lists
the Boolean equations building the model in \cite{Ste07} (communicated by the author). They exhibit
a certain degree of nondeterminism, since the functions for the \textit{off} fluents sometimes are
not the negation of the \textit{on} functions. This accounts for incomplete or inconsistent
knowledge. In the case of state transitions determined by $k$ conflicting function pairs, we
generated $2^k$ output states.

\begin{table}
\caption{Boolean rules for the nutritional stress response regulatory network, derived in
\cite{Ste07} from \cite{deJ04}. $\overline{x} \hat{=} \neg x,\: x+y \hat{=} x \vee y,\: xy \hat{=}
x \wedge y$.} \label{BoolEq}\centering
\begin{tabular}{|lcl|}
\hline
AbrB &= &SigA $\overline{\text{AbrB}}$ $\overline{\text{Spo0AP}}$\\
$\overline{\text{AbrB}}$ &= &$\overline{\text{SigA}}$ + AbrB + Spo0AP\\
SigF &= &(SigH Spo0AP $\overline{\text{SinR}})$ + (SigH Spo0AP SinI)\\
$\overline{\text{SigF}}$ &= &(SinR $\overline{\text{SinI}}$) + $\overline{\text{SigH}}$ + $\overline{\text{Spo0AP}}$\\
KinA &= &SigH $\overline{\text{Spo0AP}}$\\
$\overline{\text{KinA}}$ &= &$\overline{\text{SigH}}$ + Spo0AP\\
Spo0A &= &(SigH $\overline{\text{Spo0AP}}$) + (SigA $\overline{\text{Spo0AP}}$)\\
$\overline{\text{Spo0A}}$ &= &($\overline{\text{SigA}}$ SinR $\overline{\text{SinI}}$) +
($\overline{\text{SigH}}$ $\overline{\text{SigA}}$ ) + Spo0AP\\
Spo0AP &= &Signal Spo0A $\overline{\text{Spo0E}}$ KinA\\
$\overline{\text{Spo0AP}}$ &= &$\overline{\text{Signal}}$ + $\overline{\text{Spo0A}}$ + Spo0E + $\overline{\text{KinA}}$\\
Spo0E &= &SigA $\overline{\text{AbrB}}$\\
$\overline{\text{Spo0E}}$ &= &$\overline{\text{SigA}}$ + AbrB\\
SigH &= &SigA $\overline{\text{AbrB}}$\\
$\overline{\text{SigH}}$ &= &$\overline{\text{SigA}}$ + AbrB\\
Hpr &= &SigA AbrB $\overline{\text{Spo0AP}}$\\
$\overline{\text{Hpr}}$ &= &$\overline{\text{SigA}}$ + $\overline{\text{AbrB}}$ + Spo0AP\\
SinR &= &(SigA $\overline{\text{AbrB}}$ $\overline{\text{Hpr}}$ $\overline{\text{SinR}}$ $\overline{\text{SinI}}$ Spo0AP) +\\
&&(SigA $\overline{\text{AbrB}}$ $\overline{\text{Hpr}}$ SinR SinI Spo0AP)\\
$\overline{\text{SinR}}$ &= &$\overline{\text{SigA}}$ + AbrB + Hpr + (SinR $\overline{\text{SinI}}$) + ($\overline{\text{SinR}}$ SinI) + $\overline{\text{Spo0AP}}$)\\
SinI &= &SinR\\
$\overline{\text{SinI}}$ &= &$\overline{\text{SinR}}$\\
SigA &= &TRUE (input to the model)\\
Signal &= &TRUE or FALSE (constant, depending on the input state)\\
\hline
\end{tabular}
\end{table}

\subsection{Simulation Starting from a State Typical for the Vegetative Growth Phase}\label{simNoStress}
We performed supplementary analyses of the transitions starting from a typical state without the
starvation signal \cite[Table 4]{Ste07}. The concept lattice for the resulting transitive context
(Table \ref{tab:transitCxt}, with a part of the attribute set only) is shown in Figure
\ref{fig:noStress}. The larger circles at the bottom represent \textit{object concepts}; their
extents (highlighted part) are the four single transitions with the input state at $t=0$ or $t=2$, and the intents are
all attributes above a concept. Thus for instance, the two latter transitions have the attribute
KinA.in.on in common, designating the respective concept. Moreover, they are distinguished
unambigously from other sets of transitions by this  attribute - the concept  is \textit{generated}
by "KinA.in.on".

\begin{figure}
  \centering
  \includegraphics[width=12cm]{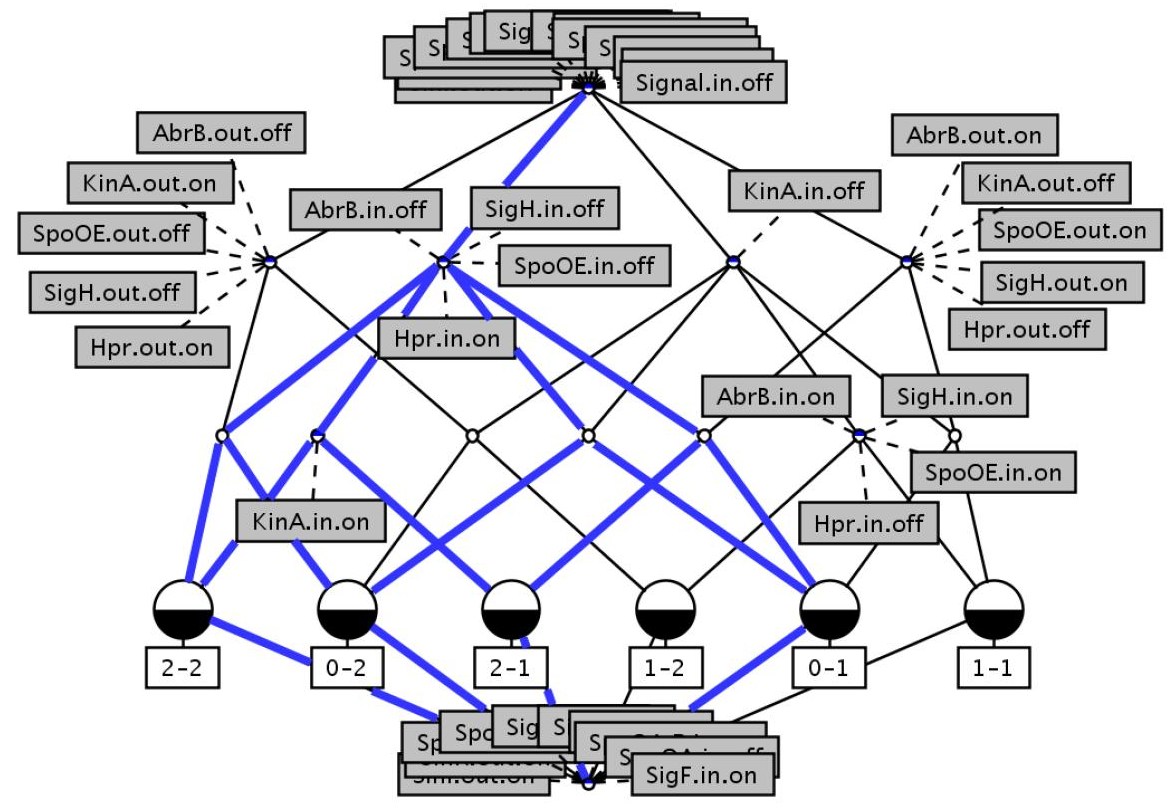}
  \caption{Concept lattice (computed and drawn with Concept Explorer) representing a simulation without nutritional stress.
\textit{Signal}: starvation; \textit{AbrB, Hpr, SigA, SigF, SigH, SinR, Spo0A} (phosporylated form
\textit{Spo0AP}): transcription factors; \textit{KinA}: kinase; \textit{Spo0E}: phosphatase;
\textit{SinI} inactivates \textit{SinR} by binding to it. $i$-$j$ indicates a transition
($\varphi_i^{in}, \varphi_j^{out})$. \textit{Bold / blue lines}: Filter (superconcepts) and ideal
(subconcepts) of the concept (\{($\varphi_0^{in}, \varphi_1^{out}), (\varphi_0^{in},
\varphi_2^{out}), (\varphi_2^{in}, \varphi_1^{out}), (\varphi_2^{in}, \varphi_2^{out})$\},
\{AbrB.in.off, SigH.in.off, SpoOE.in.off, Hpr.in.on\})} \label{fig:noStress}
\end{figure}

Implications of the stem base can be read from the lattice. For instance there are implications
between the generators of a concept:
\begin{equation}
<4> \text{ AbrB.in.off $\rightarrow$ SigH.in.off, SpoOE.in.off, Hpr.in.on}
\end{equation}
Analogous implications hold for the attributes of the conclusion, and there are implications
between attributes of sub- and superconcepts. $< 4>$ indicates that the rule is supported by four
transitions.

The bottom concept has an empty extent. Its intent is the set of attributes never occuring during
this small simulation. The top concept does not have an empty intent - as it is often the case -,
but it consists of 10 attributes common to all 6 transitions. The corresponding rule has an empty
body ($\top$):
\begin{equation}
\begin{aligned}
< 6 >\: \top \rightarrow\: &\text{Signal.in.off, SigA.in.on, SigF.in.off, Spo0A.in.on, Spo0AP.in.off,}\\
&\text{SinR.in.off, SinI.in.off, Signal.out.off, SigA.out.on, SigF.out.off, }\\
&\text{Spo0A.out.on, Spo0AP.out.off SinR.out.off SinI.out.off}
\end{aligned}
\end{equation}

Related to the simulation in the presence of nutritional stress, the transitive context has about
20 transitions, 500 concepts and 50 implications. In a such case it is more convenient to query the
implicational knowledge base. But also for the visualization of large concept hierarchies, there
exist more sophisticated tools like the ToscanaJ suite
[\texttt{http://sourceforge.net/projects/toscanaj/}].

\subsection{Analysis of All Possible Transitions}\label{subseq:allTrans}
In order to analyse the dynamics of the \textit{B. subtilis} network exhaustively, we generated
4224 transitions from all possible $2^{12}=4096$ initial states (thus the rules are nearly
deterministic). There were 11.700 transitions in the transitive context, from which we computed the
stem base containing 524 implications with support $>$ 0, but $11.023.494 \approx 2^{24}$ concepts.

It was not feasible to provide biological evidence for a larger part of the implications, within
the scope of this methodological study. This could be done by literature search, especially
automatic text mining, by new specialized experiments, or - in a faster, but less reliable way - by
comparison with high-throughput observed time series \cite[3.2]{Wol07}. Instead we will give
examples for classes of implications that can be validated or falsified during attribute
exploration in specific ways. We start with the examples of \cite[4.3]{Ste07}.
\begin{itemize}
  \item "For example, we know that in the absence of nutritional stress, sporulation should never be initiated \cite{deJ04}. We can use model checking to show this holds in our model by proving
that no reachable state exists with SigF = 1 starting from any initial state in which Signal = 0,
SigF = 0 and Spo0AP = 0." \cite[341]{Ste07} This is equivalent to the rule following from the stem
base:
\begin{equation}\label{imp:SigF.off}
  \text{Signal.in.off, SigF.in.off, Spo0AP.in.off} \rightarrow \text{SigF.out.off,}
\end{equation}
  \item SigF.out.on $\rightarrow$ KinA.out.off, Spo0A.out.off, Hpr.out.off, AbrB.out.off:\\
Spo0AP is reported to activate the production of SigF but also repress its own production (mutual
exclusion). \cite{deJ04}
  \item SigH.out.off $\rightarrow$ AbrB.out.off, SpoOE.out.off, SinR.out.off, SinI.out.off\\
All these genes are regulated $\overline{gene.out} =  \overline{\text{SigA.in}} + \text{AbrB.in}$ (+ ...).\\
\end{itemize}

In our approach, such dependencies and mutual exclusions can be checked systematically. We searched
the stem base for further interesting and simple implications:
\begin{align}
&< 4500 > \text{Spo0AP.in.on, KinA.out.off} \rightarrow \text{Hpr.out.off}\\
&< 4212 > \text{SigH.in.on. KinA.out.off} \rightarrow \text{Hpr.out.off}\\
&< 3972 > \text{AbrB.in.off, KinA.out.off} \rightarrow \text{Hpr.out.off}
\end{align}
$\overline{\text{Hpr}}$ and $\overline{\text{KinA}}$ are determined by different Boolean functions,
but they are coregulated in all states emerging from any input state characterized by the single
attributes Spo0AP.on, SigH.on or AbrB.on.
\begin{equation}\label{eq:mainRule}
\begin{aligned}
<3904> \text{AbrB.out.on } \rightarrow &\text{ SigA.in.on,
SigA.out.on, SigF.out.off,}\\
&\text{ Spo0A.out.on, Spo0E.out.on, SigH.out.on,}\\
&\text{ Hpr.out.off, SinR.out.off, SinI.out.off}
\end{aligned}
\end{equation}
AbrB is an important "marker" for the regulation of many genes, which is understandable from the
Boolean rules with hindsight. By a PubMed query, a confirmation was found for downregulation of
SigF together with upregulation of AbrB
\cite{Tom03}.

Finally we entered sets of interesting attributes as facts into the PROLOG knowledge base, such
that a derived implication was computed\footnote{and accordingly the closure of the attribute set}.
Complementary to (\ref{imp:SigF.off}), we searched after conditions for the switch towards
sporulation (SigF.out.on) and found the implication
\begin{equation}
\begin{aligned}
&\text{SigF.in.off, Spo0AP.in.off, SigF.out.on }\\
\rightarrow\: &\text{Signal.in.on. Signal.out.on, SigA.in.on, SigA.out.on, Spo0AP.out.off,}\\
&\text{Spo0A.out.off, AbrB.out.off, KinA.out.off, Hpr.out.off.}
\end{aligned}
\end{equation}

The latter four attributes follow immediately from the Boolean rules, but Spo0AP.out.off depends in
a more complex manner on the premises. It is also noteworthy that the class of input states
developing to a state with attribute SigF.out.on is only characterized by the common attributes
Signal.in.on and SigA.in.on, i.e. the initial presence or absence of no other gene is necessary for
the initiation of sporulation\footnote{For this complete simulation, the conditions Signal = SigA =
TRUE had been dropped, but they were supposed to be constant.}

\section{Discussion}\label{conclusion}
The present work translates observations and simulations of discrete temporal transitions into the
language of formal concept analysis. The application to a well studied gene regulatory network
showed how a model can be validated in a systematic way, by drawing clear and complete consequences
from the theory (the knowledge based network), and we found interesting new transition rules. The
approach could be expanded by accounting for the change of the network structure itself in strongly
different biological situations, e.g. with or without stress. Thus in ongoing work we adapt a
literature based network to observed transcriptome time series, resulting in two sets of Boolean
functions related to the stimulation of human fibroblast cells by the cytokines Tnf$\alpha$ or
Tgf$\beta$.

Until now we have established the foundation in order to exploit manyfold mathematical results of
FCA for the analysis of gene expression dynamics and of discrete temporal transitions in general.
An important question is: How can attribute exploration be split into partial problems, in this
special case? For instance, one could focus on a specific set of genes first, which is
understandable as a \textit{scaling} \cite[Definition 28]{GW99}. Then the decomposition theory of
concept lattices will be useful, which permits an elegant description by means of the corresponding
formal contexts. \cite[Chapter 4]{GW99}

The price of the logical completeness is its computational complexity. In this regard the status of
attribute exploration has not yet fully been clarified. Computation time strongly depends on the
logical structure of the context, and there exist cases where the size of the stem base is
exponential in the size of the input \cite{Kuz06}. However, deriving an implication from the stem
base is possible in linear time, related to the size of the base, and the PROLOG queries in Section
\ref{subseq:allTrans} were very fast. As demonstrated in Section \ref{subsec:bg}, attribute
exploration can be shortened by background knowledge. Further it will be crucial to decide
implications without the necessity to generate all possible transitions. For that purpose, model
checking \cite{Esp94} could be a promising approach, or the structural and functional analysis of
Boolean networks by an adaptation of metabolic network methods in \cite{Kla06}. There, determining
activators or inhibitors corresponds to the kind of rules found by our method, and logical steady
state analysis indicates which species can be produced from the input set and which not. An
exciting direction of research would be to conclude dynamical properties of Boolean networks from
their structure and the transition functions, e.g. by regarding them as polynomial dynamical
systems over finite fields \cite[Section 4]{Lau05} and by exploiting theoretical work in the
context of cellular automata \cite[Section 6]{Lau05}.

The present work is a first step to use the potential of formal concept analysis for solving
questions within systems biology. As indicated, many directions of research are possible. We
encourage their investigation and are open to any collaboration with mathematicians, computer
scientists or (systems) biologists.

\subsubsection*{Acknowledgement.} The work was supported by the German Federal Ministry of Education
and Research BMBF (FKZ 0313652A).

\begin{footnotesize}\subsubsection*{Published in:} K. Horimoto et al. (Eds.): AB 2008, LNCS 5147, pp. 230-244.\\
\textcopyright$\:$ Springer-Verlag Berlin Heidelberg 2008                                                      \end{footnotesize}


\begin{thebibliography}{}
\bibitem{Cha04}
Chabrier-Rivier, N. et al.: Modeling and Querying Biomolecular Interaction
Networks. \textit{Theor. Comp. Sc.} {\bfseries 325}(1), 25-44 (2004)

\bibitem{Cho06}
Choi, V., et al.: Using Formal Concept Analysis for Microarray Data Comparison. \textit{Advances in Bioinformatics and Computational Biology} \textbf{5}, 57-66 (2006)

\bibitem{GSW05}
Ganter, B., Stumme, S., Wille, R.: \textit{Formal Concept Analysis -  Foundations and
Applications}. LNAI, vol. 3626. Springer, Heidelberg (2005)


\bibitem{GW99}
Ganter, B., Wille, R.: \textit{Formal Concept Analysis - Mathematical Foundations}. Springer, Heidelberg (1999)

\bibitem{deJ04}
de Jong, H., et al.: Qualitative Simulation of the Initiation of Sporulation in \textit{Bacillus subtilis}. \textit{Bulletin of Mathematical Biology} \textbf{66}, 261-299 (2004)

\bibitem{Esp94}
Esparza, J.: Model checking using net unfoldings. \textit{Sci. Comput. Programm.} \textbf{23}, 151-195. (1994)

\bibitem{Hen74}
von Hentig, H.: \textit{Magier oder Magister?} \"Uber die Einheit der Wissenschaft im
Verst\"andigungsprozess. Suhrkamp, Frankfurt (1974)

\bibitem{Kau93}
Kauffman, S.A.: \textit{The Origins of Order: Self-Organization and Selection
in Evolution}. Oxford University Press, New York (1993)

\bibitem{Kla06}
Klamt, S., et al.: A methodology for the structural and functional analysis of signaling and regulatory networks. \textit{BMC Bioinformatics} \textbf{7}(56) (2006)

\bibitem{Kuz06}
Kuznetsov, S.O., Obiedkov, S.A.: Counting Pseudo-intents and \#P-completeness. In: Missaoui, R.,
Schmid, J. (eds.) \textit{ICFCA 2006}. LNCS, vol. 3874, pp. 306-308. Springer, Heidelberg (2006)

\bibitem{KGC05}
King, R.D., Garrett, S.W, Coghill, G.M.: On the Use of Qualitative Reasoning to Simulate and
Identify Metabolic Pathways. \textit{Bioinformatics} {\bfseries 21}(9), 2017-2026 (2005)


\bibitem{Lau05}
Laubenbacher, R.: Algebraic Models in Systems Biology. In: Anai, H., Horimoto, K. (eds.)
\textit{Algebraic Biology 2005}, pp. 33-35. Universal Academy Press, Tokyo (2005)

\bibitem{Mot08}
Motameny, S., Versmold, B., Schmutzler, R.: Formal Concept Analysis for the Identification of
Combinatorial Biomarkers in Breast Cancer. In: Medina, R., Obiedkov, S.A. (eds.) \textit{ICFCA
2008}. LNCS, vol. 4933, pp. 229-240. Springer, Heidelberg (2008)

\bibitem{Pei35}
Peirce, C.S.: How to Make Our Ideas Clear. In: Hartshorne, C., Weiss, P. (eds.) \textit{Collected
papers}. Harvard University Press, Cambridge / Mass. (1931-35)

\bibitem{Rud01}
Ganter, B., Rudolph, S.: Formal Concept Analysis Methods for Dynamic Conceptual Graphs. In:
\textit{ICCS 2001}, LNAI 2120, pp. 143-156. Springer, Heidelberg (2001)

\bibitem{Ste07}
Steggles, L. J. et al.: Qualitatively modelling and analysing genetic regulatory networks: a
{P}etri net approach. \textit{Bioinformatics} {\bfseries 23}(3), 336-343 (2007)

\bibitem{Tom03}
Tomas, C.A. et al.: DNA array-based transcriptional analysis of asporogenous, nonsolventogenic
{C}lostridium acetobutylicum strains SKO1 and M5. \textit{J Bacteriol.}  \textbf{15}, 4539-47
(2003)

\bibitem{Wol07}
Wollbold, J.: Attribute Exploration of Discrete Temporal Transitions. In: G{\'e}ly, A., et al.
(eds.) \textit{Contributions to ICFCA 2007}, pp. 121-130. Clermont-Ferrand (2007)

\bibitem{Wol08a}
Wollbold, J., Huber, R., Wolff, K.E.: Conceptual Representation of Gene Expression Processes. In:
\textit{Knowledge Processing in Practice 2007}. To appear in LNAI. Springer, Heidelberg (2008)

\bibitem{Zic91}
Zickwolff, M.: \textit{Rule Exploration: First Order Logic in Formal Concept Analysis.} PhD thesis.
University of Technology, Darmstadt (1991)

\end{thebibliography}
\end{document}